\documentclass[preprint,11pt]{elsarticle}

\makeatletter
\def\ps@pprintTitle{%
 \let\@oddhead\@empty
 \let\@evenhead\@empty
 \def\@oddfoot{}%
 \let\@evenfoot\@oddfoot}
\makeatother

\usepackage[titletoc,title]{appendix}

\makeatletter
  \@addtoreset{chapter}{part}
  \@addtoreset{@ppsaveapp}{part}
\makeatother

\usepackage[latin1]{inputenc}
\usepackage[english]{babel}
\usepackage{caption}

\usepackage{xcolor}
\usepackage{amsmath,amsthm}
\usepackage{amsfonts}
\usepackage{caption}
\usepackage{relsize}
\usepackage{algorithm}
\usepackage{algcompatible}
\usepackage{setspace}
\usepackage{euscript}
\usepackage{enumerate}
\usepackage{enumitem}

\newtheorem{theo}{Theorem}

\begin{document}
\renewcommand{\baselinestretch}{1.0}

\pagestyle{plain}

\begin{frontmatter}

\title{An improved solution approach for the Budget constrained Fuel Treatment Scheduling problem}
\date{}

 \author[1,2]{Federico Della Croce}
  \author[1]{Marco Ghirardi}
 \author[1]{Rosario Scatamacchia}
 
  \address[1]{\small Dipartimento di Ingegneria Gestionale e della Produzione, Politecnico di Torino,\\ Corso Duca degli Abruzzi 24, 10129 Torino, Italy, \\{\tt \{federico.dellacroce, marco.ghirardi, rosario.scatamacchia\}@polito.it }}
 \address[2]{CNR, IEIIT, Torino, Italy}
 
\begin{abstract}
This paper considers
the budget constrained fuel treatment scheduling (BFTS) problem where,
in the context of wildfire mitigation, 
the goal is to inhibit the potential of fire spread in a landscape by proper fuel treatment activities. Given a time horizon represented by consecutive unit periods, the landscape is divided into cells and represented as a grid graph where each cell has a fuel age that increases over time and becomes old if no treatment is applied in the meantime: this induces a potential high fire risk whenever two contiguous cells are old. Cells fuel ages can be reset to zero under appropriate fuel treatments but there is a limited budget for treatment in each period. The problem calls for finding a suitable selection of cells to be treated so as to minimize the presence of old contiguous cells over the whole time horizon. We prove that problem 
BFTS is strongly $NP$-complete on paths and thus on grid graphs and show that no polynomial time approximation algorithm exists unless 
$\mathcal{P} = \mathcal{NP}$.
We provide an enhanced integer linear programming  formulation of the problem with respect to the relevant literature that shows up to be efficiently solved by an ILP solver on reasonably large size instances. Finally, we consider a harder periodic variant of the problem with the aim of finding  a cyclic treatment plan with cycles of length T and propose a matheuristic approach capable of efficiently tackling those instances where an ILP solver applied to the ILP formulation runs into difficulties.
\end{abstract}


\end{frontmatter}

\section{Introduction}
One of the most studied problems in wildfire mitigation is Fuel Management. Fuel Management aims to reduce potential fire intensity and fire spread in vegetation areas. Recent interest in fire and fuel management is particularly motivated by short fire return intervals and new zones where fire was excluded during the 20th century. This creates a need for the long-term reduction of fuel loads. The particularity of this management problem is that it requires new and spatially explicit management science methods.
The objective of Fuel Management is the modification of potential fire behavior or fire effects by undertaking a minimum action on its fuel. The main idea consists in applying a treatment on selected specific areas, either by harvesting or burning according to some ranking of risk or effectiveness. In particular, several studies recognized the effectiveness of prescribed burning in reducing fire intensity and severity of wildfires (see, e.g., \cite{BOAL09}, \cite{FEBO03}).
Several fuel management optimization problems and related solution techniques were considered in recent years.
We mention among others one of the first papers in literature
applying mixed integer linear programming (MILP) in the context of wildfire fuel management \cite{WRK08}. Since then, many operations research approaches have been proposed applied to schedule fuel treatment activities. We refer to the surveys in \cite{MHH12} and \cite{GMP17} on the matter. 
Recently, a MILP-based multi-period optimization framework 
for prescribed burning activities over a finite planning horizon
was proposed in \cite{MPZ18}.
In \cite{BMMRG19},
a stochastic programming approach for optimizing
fuel reduction treatment allocation was introduced.
A strongly investigated problem is the so-called 
Budget constrained Fuel Treatment Scheduling (BFTS) problem.
This problem was firstly tackled by means of MILP modeling in \cite{MiHeMa} and gave rise to several publications on related generalized models in \cite{leReHe}, \cite{MiHeMa}, \cite{RaOzHeRe} and \cite{16RaOzHeRe}. 
Some complexity and approximation findings  on fuel treatment scheduling on graphs were proposed in \cite{DeTa15}. 
Our paper focuses on the original BFTS problem 
which is a multi-period problem where the goal is to inhibit the potential of fire spread in a landscape by proper fuel treatment activities (such as prescribed burning). The landscape can be divided in cells for fuel treatments according to some features of interest (e.g. type of vegetation) and/or operational conditions. Each cell has a fuel age that increases over time. A cell becomes ``old'' if its fuel age gets larger than the inhibition period caused by a treatment. High fire risk may occur in the landscape whenever two contiguous vegetation areas are old.  When a fuel treatment occurs, the fuel age of a cell is reset to zero. Also, fuel treatment activities have a cost and there is an available limited budget for treatments in each period. The problem calls for finding a suitable selection of the areas to be treated so as to minimize the presence of old contiguous areas over the whole time horizon.
This problem can be typically represented by means of grid graphs where vertices represent areas and edges represent the connection between adjacent areas. Two areas might be connected not only in space but also according to other conditions such as prevailing wind directions. 
To the authors' knowledge, the complexity status of BFTS on grid graphs
is still open. Correspondingly, we show that the decision version of BFTS on paths
(and henceforth also on grid graphs) is strongly $NP$-complete. Further, we show that no polynomial time approximation algorithm exists for BFTS on paths unless $\mathcal{P} = \mathcal{NP}$.
 Next we focus on the practical solvability of the problem.
In the ILP formulations in \cite{leReHe} and \cite{MiHeMa}, variables $X_{it}$ indicate if area $i$ is or not treated at time period $t$ and variables $O_{it}$ indicate if area $i$ is or not old at time period $t$. In these formulations, however, big-M constraints are used. This may strongly limit the performance of MILP solvers in reaching optimal/suboptimal solutions. 
We show that a pair of alternative and much more simplified ILP formulations exists that avoids the need of big-M constraints. The best of these formulations when tested by means of MILP solver CPLEX 12.8 allows to optimally solve reasonably large size instances of the problem within limited CPU time.

We consider then a new harder problem variant with periodic planning (hereafter denoted PBFTS) where the goal is to find a cyclic treatment plan with cycles of length $T$ starting from an initial state to be determined. On this problem variant, CPLEX 12.8 applied to the proposed ILP starts running into difficulties with instances with $400$ cells.
Correspondingly, a matheuristic approach is proposed in order to deal with larger size problems.
Matheuristics are methods that attracted the attention of the community of researchers (see for instance \cite{Ball11,DCGS13}), giving rise to an impressive amount of research in recent years. We mention applications of matheuristics on routing \cite{MLGPP19,SAS20},
 packing \cite{BDCG15,MABRT17}, rostering \cite{DCS14,DNV18}
 and machine scheduling \cite{DCGS14,DCGS19,FPR17} just to cite a
 few of them.
Matheuristics rely on the general idea of exploiting the strength of both metaheuristic algorithms and exact methods. Here, the proposed matheuristic algorithm is based on an overarching neighborhood search approach with an intensification search phase realized by a MILP solver. 
The proposed MILP formulation for BFTS and PBFTS 
constitutes a backbone model generalizable also
to the variants presented  in \cite{leReHe}, \cite{MiHeMa}, \cite{RaOzHeRe} and \cite{16RaOzHeRe}. 
We cite, among others, the possible request of having no adjacent areas treatment in the same time period and/or minimal/maximal periods of time between two treatments of the same areas. The effectiveness of the developed models and algorithms is evaluated on a representative set of instances with up to $1225$ cells.

The paper proceeds as follows.
Section \ref{Sec:notat} provides the relevant notation for BFTS and discusses complexity and approximation issues.  Section \ref{Sec:formulations} presents enhanced ILP formulations both for BFTS and its periodic variant. Section \ref{Sec:Matheuristic}
presents the proposed matheuristic approach and Section \ref{SecTests}
provides the relevant related computational results both on BFTS and PBFTS. Section \ref{Sec:concl} concludes the paper with final remarks.

\section{Notation and complexity of BFTS}
\label{Sec:notat}
In BFTS, a set $A$ of areas (or cells) to be treated and a time horizon $T$ are given. For each area $i$, we denote by $\phi_i$ the set of its adjacent areas, by $a_i$  its initial fuel age and by $\sigma_i$ its fuel age threshold value. The threshold $\sigma_i$ defines the number of periods, starting from the last treatment, after which area $i$ becomes old. In each period $t = 1, \dots, T$,  there are treatment costs for each area $i$ denoted by $c_{it}$ and a budget on fuel treatments denoted by $b_t$. Finally, we denote by $w_{tij}$ the cost associated with the presence of two old adjacent areas $i$ and $j$ in period $t$. The problem calls for minimizing the weighted sum of old contiguous areas over the whole time horizon by a proper selection of the areas to treat without exceeding the available budget in each period.\\

As BFTS is typically of interest for problems that can be represented by means of grid graphs, we are interested in determining the computational 
complexity of the BFTS problem on this specific class of graphs or any related subclass. 
Let denote by BFTS-B the decision version of problem BFTS where we search for a feasible solution of BFTS 
so that the weighted sum of old contiguous areas does not exceed a given bound $B$ within the whole time horizon T. 
We prove that  BFTS-B is strongly $NP$-complete on paths (and consequently also on grid graphs) by reduction from 
the 3-PARTITION problem below which is well known to be $NP$-complete in the strong sense \cite{GJ82}.\\

\noindent {\bf 3-PARTITION} \\
INSTANCE:
Positive integers $n,K$ and a set of integers 
$S = \{s_1,s_2, \cdots, s_{3n} \}$ 
with $\sum_{i=1}^{3n}s_i = nK$ and $\frac{K}{4} < s_i < \frac{K}{2} $ for $i=1,\cdots,3n$. \\
QUESTION:
Does there exist a partition $< S_1, S_2, \cdots, S_n >$ of
S into 3-elements sets such that,
for each $j$, $\sum_{s_i \in S_j}s_i =K$? 

\begin{theo}
\label{strongNP}
BFTS-B  on paths is strongly $NP$-complete.
\end{theo}

\begin{proof}
We focus on a special case of problem BFTS-B where we search for a solution with no old contiguous areas, 
that is with $B = 0$. W.l.o.g we assume that all weights $w_{tij}$ are one. Also, we have a path with $6n-1$ vertices and consider $T=n$.
The path is constituted by 
$3n$ vertices $u_1, \cdots, u_{3n}$ and $3n-1$ vertices $v_1, \cdots, v_{3n-1}$ sequenced as follows: \\
\[u_1 - v_1 - u_2 - v_2 - \cdots -  u_i - v_i - \cdots - u_{3n-1} - v_{3n-1} - u_{3n}.\]
\\
All the $6n-1$ vertices have zero initial age and an identical threshold $\sigma$ equal to $n-1$. All $u_i$ vertices correspond to elements $s_i \in S$ of 3-PARTITION. These vertices have a treatment cost in each period equal to $s_i$.
The budget in each period is $b=\frac{\sum_{i=1}^{3n}s_i}{n}= K$.
All $v_i$ vertices have cost $b+1$, thus these vertices can never be treated. 
Hence, to find a feasible solution with no old contiguous areas, 
all $u_i$ vertices must be treated in one of the periods $(1, \cdots, n)$ so that at period $T$ when all $v_i$ vertices are old, 
all $u_i$ vertices are not old. Correspondingly, the sum of the treatment costs over the $n$ periods is equal to the sum of all the $s_i$ elements in 3-PARTITION, i.e. $\sum_{i=1}^{3n}s_i = nK$.
Besides, the sum of the treatment costs in each period cannot exceed $K$. These conditions imply that this special case of BFTS-B has a solution if and only if 3-PARTITION has a solution: a solution of 3-PARTITION gives a solution of BFTS-B where the elements $s_i$ in each subsets $S_1, \cdots,  S_n$ indicate the associated vertices $u_i$ to be treated in each period $1, \cdots, n$. Likewise, as $\frac{K}{4} < s_i < \frac{K}{2} $ for $i=1,\cdots,3n$, a solution of BFTS-B must treat exactly three vertices in each period with treatment costs equal to $K$, thus providing a solution of 3-PARTITION. 
\end{proof}
\medskip 

In terms of approximability, the strong NP-completeness result of Theorem \ref{strongNP} rules out the existence of a Fully Polynomial Time Scheme (FPTAS) for BFTS even on paths. However, we can prove a more general result that no polynomial time approximation algorithm exists for BFTS on paths unless $\mathcal{P} = \mathcal{NP}$.\\ 
We say that an approximation algorithm for BFTS has a finite approximation ratio $\rho \geq 1$ if it provides a solution value that is not larger than the product between $\rho$ and the optimal solution value in any instance. Also, using the notation for 3-PARTITION, we recall the well known $NP$-complete PARTITION problem \cite{GJ82} below.
\medskip

\noindent {\bf PARTITION} \\
\noindent INSTANCE: Finite set $S$ of $n$ positive integers  $s_1, s_2, \cdots, s_n$. \\
QUESTION: Is there a subset $S' \subseteq S$ such that $\sum_{s_i \in S'} s_i =  \sum_{s_i \in S \backslash S'} s_i = \frac{\sum_{i=1}^{n}s_i}{2}$?
\medskip

By using a reasoning similar to the one in the proof of Theorem \ref{strongNP}, we state the following theorem.
\begin{theo}
\label{TheoInapprox}
No polynomial time approximation algorithm exists for BFTS on paths unless $\mathcal{P} = \mathcal{NP}$.
\end{theo}

\begin{proof}
 For a given instance of PARTITION,  we consider an instance of BFTS with unit weights $w_{tij}$ and $T=2$.
We consider a path with $n$ vertices $u_1, \cdots, u_{n}$ and $n-1$ vertices $v_1, \cdots, v_{n-1}$ connected as follows:\\
\[u_1 - v_1 - u_2 - v_2 - \cdots -  u_i - v_i - \cdots - u_{n-1} - v_{n-1} - u_{n}.\]
\\
All the $2n-1$ vertices have zero initial age and a unit fuel threshold. All $u_i$ vertices have a treatment cost equal to $s_i$ in each period, thus they correspond to elements $s_i \in S$ of the instance of PARTITION. The budget in each period is $b=\frac{\sum_{i=1}^{n}s_i}{2}$. All $v_i$ vertices have treatment cost $b+1$, so they cannot be treated and become old in the second period.\\
If the given instance of PARTITION has a solution, we can treat the vertices $u_i$ associated with elements $s_i \in S'$ in the first period, the remaining vertices $u_i$ in the second period (as in both periods the budget is $\frac{\sum_{i=1}^{n}s_i}{2}$) and obtain a solution for BFTS with no old contiguous areas and zero objective value.\\
Else, any optimal solution of the BFTS instance would induce at least two old contiguous areas, implying a positive objective function. Hence, a polynomial time approximation algorithm for BFTS would allow us to decide the PARTITION problem by checking if the approximate solution of BFTS is strictly positive. Obviously, this is not possible unless $\mathcal{P} = \mathcal{NP}$.
\end{proof}

Notice that the inapproximability result of Theorem \ref{TheoInapprox} already applies for BFTS on paths with $T=2$ and unit $w_{tij}$. 

\section{Enhanced ILP formulations for BFTS}
\label{Sec:formulations}

We focus now on the practical solvability of BFTS. We introduce the following simplified ILP formulation that avoids the use of big-M constraints as in the model proposed in \cite{MiHeMa}. We consider binary variables $x_{ti}$ equal to 1 iff area $i \in A$ is treated at time $t$ and binary variables $q_{tij}$ equal to 1 iff two adjacent areas $i$ and $j$ are old in period $t$. Correspondingly, we obtain the following model denoted as $M_1^{BFTS}$:

\begin{align}
&M_1^{BFTS}: &  \nonumber \\ 
\min &  \sum_{t = 1}^T  \sum_{i \in A} \sum_{j \in \phi_i} w_{tij} q_{tij}
      \label{eq:obj} \\
\mbox{s.t.}  \nonumber \\
&   \sum\limits_{\substack{p = \max \\ \{1; t - \sigma_i \}}}^{t} x_{pi} +  \sum\limits_{\substack{p = \max \\ \{1; t - \sigma_j \}}}^{t} x_{pj} + q_{tij} \geq 1 &  \forall i \in A, j \in \phi_i, t = 1,\dots,T: \nonumber  \\
& &  t > \max \{ \sigma_i - a_i,\; \sigma_ j - a_j \} \label{eq:Cover}  \\
&  \sum_{i \in A}  c_{ti}x_{ti} \leq b_t & t = 1, \dots, T  \label{eq:Budget}  \\   
&  x_{ti} \in \{0,1\}       &  \forall i \in A, t = 1, \dots, T  \label{eq:Xdomain} \\
&  q_{tij} \in \{0,1\}       &  \forall i \in A, j \in \phi_i, t = 1, \dots, T  \label{eq:Qdomain} 
\end{align}

The objective function \eqref{eq:obj} minimizes the weighted sum of old adjacent areas over the whole time horizon. Constraints \eqref{eq:Cover} detect the presence of old adjacent cells $i$, $j$ in period $t$: if in the relevant periods both areas $i$ and $j$ are not treated (i.e. all the related variables $x$ are equal to zero), then we have $q_{tij} = 1$; else we have $q_{tij} = 0$ given the objective function \eqref{eq:obj}. Constraints \eqref{eq:Budget} represent the budget constraints for each period $t=1, \dots, T$. Constraints \eqref{eq:Xdomain} and \eqref{eq:Qdomain} define the domain of the variables. We point out that variables $q_{tij}$ could be defined as nonnegative continuous variables, namely $q_{tij} \geq 0$, as variables $q_{tij}$ are set to either 0 or 1 in any optimal solution due to constraints \eqref{eq:Cover} and the objective function \eqref{eq:obj}. Notice also that model  $M_1^{BFTS}$ can be seen as a variant of the set covering problem.
This formulation considerably enhances the MILP model for BFTS proposed in \cite{MiHeMa}. As illustrated in our computational tests (see Section \ref{SecTests}), a MILP solver launched on model $M_1^{BFTS}$ already provides better performances than the reference model from the literature.\\

\smallskip
In model $M_1^{BFTS}$ , we notice that, for any given area $k$, the sum $\sum\limits_{p = \max \{1; t - \sigma_k \}}^{t} x_{pk}$ has to be repeated each time the area is considered in constraints \eqref{eq:Cover}. This may induce a large number of nonzero coefficients in the constraint matrix whenever the value of each threshold $\sigma_i$ and the overall number of adjacent areas are not small. Generally speaking, a large number of nonzero coefficients may affect the performance of a MILP solver launched on model $M_1^{BFTS}$. Hence, we propose an alternative ILP formulation where we replace each sum $\sum\limits_{p = \max \{1; t - \sigma_k \}}^{t} x_{pk}$ in constraints \eqref{eq:Cover} with one auxiliary binary variable $y_{tk}$ equal to 1 if area $k$ is young at time $t$ or equal to 0 if area $k$ is old at time $t$. This condition is ensured by introducing an additional constraint $\sum\limits_{p = \max \{1; t - \sigma_k \}}^{t} x_{pk} - y_{tk} \geq 0$. The constraint states if no treatment occurs for area $k$ in the relevant periods up to $t$, then the area will be old, i.e. $y_{tk} = 0$. Correspondingly, we obtain the following model $M_2^{BFTS}$:
\begin{align}
&M_2^{BFTS}:		& \nonumber \\ 
\min 
&  \sum_{t = 1}^T  \sum_{i \in A} \sum_{j \in \phi_i} w_{tij} q_{tij}
      \label{eq:obj2} \\
\mbox{s.t.} \nonumber \\
& \quad y_{ti} + y_{tj} + q_{tij} \geq 1    &  \forall i \in A, j \in \phi_i, t = 1,\dots,T: \nonumber  \\
& &  t > \max \{ \sigma_i - a_i,\; \sigma_ j - a_j\} \label{eq:Cover2}  \\
& \sum\limits_{\substack{p = \max \\ \{1; t - \sigma_i \}}}^{t} x_{pi} - y_{ti} \geq 0&  \forall i \in A, t = 1,\dots,T: t > \sigma_i - a_i \label{eq:theo2} \\
&  \sum_{i \in A}  c_{ti}x_{ti} \leq b_t & t = 1,\dots,T \label{eq:Budget2}  \\   
&  y_{ti}, x_{ti} \in \{0,1\}  & \forall i \in A, t = 1,\dots,T \label{eq:Xdomain2} \\
&  q_{tij} \in \{0,1\}       & \forall i \in A, j \in \phi_i, t = 1,\dots,T \label{eq:Qdomain2} 
\end{align}
The computational tests on reference instances from the literature (see Section \ref{SecTests}) show that model $M_2^{BFTS}$ provides better performances than model  $M_1^{BFTS}$. \\

\smallskip

For the sake of exposition, from now on we refer only to model $M_2^{BFTS}$ in the description of problem variants and of the matheuristic approach presented in Section \ref{Sec:Matheuristic}. However, we remark that the same analysis can be easily extended to model $M_1^{BFTS}$ as well. 

\subsection*{Column generation}
We also investigated an ILP formulation with an exponential number of variables and evaluated a corresponding column generation approach (see, e.g., \cite{CG1} for an introduction on column generation). We decomposed model $M_2^{BFTS}$ by considering all possible feasible treatment plans of the areas in each period $t= 1, \dots, T$. This induced a number of plans that was exponential in the number of areas. Correspondingly, for each period $t$ we defined a pricing problem to compute feasible treatment plans according to the budget constraints. The choice of a derived treatment plan $s$ in period $t$ was associated with a binary variable $z_{s}^t$ in a master problem containing also variables $y_{ti}$, $q_{tij}$ and without variables $x_{ti}$. After some preliminary computational tests, we noticed that solving the linear relaxation of the master problem did not require negligible computational times even in small instances. The presence of variables $y_{ti}$, $q_{tij}$ and related objective function and constraints turned out to slow down the iterative column generation process. In the light of this, the use of such a column generation approach and possible related branch and price schemes does not seem a viable option for BFTS.

\subsection{Periodic BFTS}
From a practical point of view, it might be of interest for the decisions makers to define periodic policies for fuel management. This motivates us to consider a new problem variant with periodic planning, here denoted as PBFTS. The goal is to find a continuously repeatable treatment plan with cycles of length $T$ starting from an initial state to be determined. In this variant we assume that the initial ages $a_i$ are set to zero. To handle periodicity, we replace constraints \eqref{eq:theo2} in model $M_2^{BFTS}$ with the following constraints:
\begin{align}
& \sum\limits_{p = t - \sigma_i }^{t} x_{pi} - y_{ti} \geq 0  &  \forall i \in A, t = 1,\dots,T: t - \sigma_i \geq 1  \label{eq:theo3} \\
& \sum\limits_{p = 1}^{t} x_{pi} + \sum\limits_{p = T + t - \sigma_i }^{T} x_{pi} - y_{ti} \geq 0&  \forall i \in A, t = 1,\dots,T: t - \sigma_i < 1  \label{eq:theo4} 
\end{align}
and denote the corresponding model as $M_2^{PBFTS}$. Similarly to  constraints \eqref{eq:theo2}, constraints \eqref{eq:theo3} state that if $t > \sigma_i$, then the fact that area $i$ is old or not in period $t$ depends on the presence of treatments in periods $t-\sigma_i, \dots, t$. If instead $t \leq \sigma_i$, also periods from the preceding cycle have to be considered according to constraints \eqref{eq:theo4}.

\section{A matheuristic approach}
\label{Sec:Matheuristic}
Computational tests on problem PBFTS (see Section \ref{SecTests}) show that the addition of periodicity makes the problem much more challenging so that already for several instances with $100$ cells, the considered ILP solver (CPLEX 12.8) applied to model $M_2^{PBFTS}$ fails to reach the optimal solution within $1800$ seconds of CPU time denoting the need of a heuristic algorithm. \\

We designed a matheuristic approach that relies on computing an initial solution from the optimal solution of the linear relaxation of $M_2^{PBFTS}$. Then, a local search procedure is applied to improve the incumbent solution. Both the initial solution and the improvement procedure can be also applied to BTFS without major modifications: what changes is just the underlying model used by the algorithms. \\ 

The main idea of the initial solution procedure (denoted as Algorithm \ref{Algo:InitialSolution})  is to progressively constraint to 0 or 1 the variables $x_{it}$, until all of them have been fixed. The algorithm begins from the first time period and fixes at most $K$ fractional variables (the ones with the highest value in the current solution), without exceeding the budget of the period. Then, the optimal solution of the model linear relaxation is updated, and the cycle is repeated. The current time period is increased each time all its variables have been fixed to an integer value. 

The detailed steps of the procedure are summarized as follows. The inputs (step 1) are the value of $K$ (maximum number of variables fixed at each iteration), and the underlying model $M$ (the linear relaxation of $M_2^{BFTS}$ or $M_2^{PBFTS}$). At first, the current time period $t_{curr}$ is set to the first one, the remaining usable budget $B_t$ is initialized to $b_t$ for each time period, and the sets $\Phi_0$ and $\Phi_1$ of indices $(t,i)$ of the $x_{it}$  variables constrained to $0$ or $1$ is initialized to the empty set (Step 2). Then, a first linear relaxation of model $M$ is performed, retrieving the initial solution $\bar{X}$ (step 3). The main cycle (step 4) is performed until the current solution $\bar{X}$ is not integer. Steps 5-13 fix to $0$ or $1$ (by adding their indices to $\Phi_0$ or $\Phi_1$) any variable $\bar{x}_{it}$ in solution $\bar{X}$ not yet constrained and equal to $0$ and $1$. In case a variable is set to $1$, the remaining budget of the period is updated accordingly (step 11). Cycle 14-26 fixes to $1$ at most $K$ variables $x_{it}$ in the current time period $t_{curr}$, considering them in order of non increasing (fractional) value. In particular, step 15 identifies the highest fractional value variable, which is added to $\Phi_1$ if the remaining budget of the current time period allows it (steps 20-22). All variables whose cost exceed the current remaining budget are set to 0 (steps 23-24), and the cycle is ended when all variables of time period $t_{curr}$ have been inserted in $\Phi_0$ or $\Phi_1$ (steps 16-18). Before ending the main cycle, the model with the additional constraints on the fixed variables is run, and the current solution updated. \\

\begin{algorithm}
	\caption{Initial solution.}
	\label{Algo:InitialSolution}
	\begin{algorithmic}[1]
		\STATE INPUT: The maximum number of variables fixed at each iteration $K$, model $M$ (linear relaxation of $M_2^{BFTS}$ or $M_2^{PBFTS}$)
		\STATE Set $t_{curr}=1; \Phi_0 = \Phi_1 = \emptyset; B_t = b_t, t \in 1, \cdots,T$
		\STATE Solve the linear relaxation of model $M$. Let $\bar{X}$ be the solution.
		\WHILE{solution $\bar{X}$ is not integer} 
		\FORALL {$(t,i)$ indices such that $t \geq t_{curr}; (t,i) \notin \Phi_0; (t,i) \notin \Phi_1$} 
		\IF {$\bar{x}_{ti} = 0$} 
		\STATE $\Phi_0 = \Phi_0 + (t,i)$
		\ENDIF
		\IF {$\bar{x}_{ti} = 1$} 
		\STATE $\Phi_1 = \Phi_1 + (t,i)$
		\STATE $B_t = B_t - c_{ti}$
		\ENDIF
		\ENDFOR
		\FOR{$count=1,\cdots,K$}
		\STATE Let $i$ be the index corresponding to the the highest $\bar{x}_{t_{curr},i}$ such that $(t_{curr},i) \notin \Phi_0, (t_{curr},i) \notin \Phi_1$.
		\IF {there is no such an index}
		\STATE $t_{curr}=t_{curr}+1$; 
		\STATE \textbf{break for}
		\ENDIF
		\IF{$B_{t_{curr}} > c_{t_{curr},i}$}
		\STATE $\Phi_1 = \Phi_1 + (t_{curr}, i)$
		\STATE $B_{t_{curr}} = B_{t_{curr}} - c_{t_{curr},i}$
		\ELSE
		\STATE $\Phi_0 = \Phi_0 + (t_{curr}, i)$		
		\ENDIF
		\ENDFOR
		\STATE Solve again the linear relaxation of model $M$, adding constraints $x_{ti} = 0$ for all $(t,i) \in \Phi_0$ and $x_{ti} = 1$ for all $(t,i) \in \Phi_1$. Let again $\bar{X}$ be its optimal solution.
		\ENDWHILE
		\STATE OUTPUT: the feasible (integer) solution $\bar{X}$.
	\end{algorithmic}
\end{algorithm}

The matheuristic improvement procedure (denoted as Algorithm \ref{Algo:MatheuristicAlgorithm}) starts with the solution found by the previous algorithm as current solution, and iteratively improves it with a scheme based on the neighborhood search approach. Each iteration explores the neighborhood by constructing a problem where the variables to be optimized refers to a subset of time periods, while other ones are fixed to the value they have in the current solution. The neighborhood size is variable, beginning with a smaller number of periods, increased when a complete exploration of the neighborhood has been performed without finding any improvement.

More in details, the procedure  requires the underlying ILP model ($M_2^{BFTS}$ or $M_2^{PBFTS}$), an initial solution $\bar{X}$, and two settings: the initial and final neighborhood size $T_{init}$ and $T_{final}$ (step 1). 
The main cycle (steps 3-21) reoptimizes at each iteration a part of the current solution, creating a problem with a consecutive number of $T_{dim} < T$ periods, with $T_{dim}$ ranging from $T_{init}$ to $T_{final}$. It begins setting the initial period to be reoptimized $t_{min}$ to 1 and $t_{impr}$, the period when the neighborhood exploration should end, to $T$ (step 3). Then, an internal cycle is repeated while $t_{min} \neq t_{impr}$ (steps 4-20). Here, the model for the neighborhood reoptimization is built: the variables to be reoptimized are all the $x_{it}$ with $t$ in the interval $[t_{min}, t_{min}+T_{dim}-1]$, but when the second limit exceed $T$ the reoptimization interval resumes from time $1$ becoming $[t_{min}, T] \cup [1, t_{min}+T_{dim}-T-1]$, implementing a cycling neighborhood. All the variables not belonging to the reoptimization interval are fixed to the value they already have in the current solution (steps 5-10). The model is then optimally solved, and its optimal objective function $OF(\bar{x})$ retrieved. If the new solution is better than the previous one, $t_{impr}$ is updated (steps 11-15). Note that at each iteration the optimal solution of the reoptimization problem can't be worse than the current one (the latter is still a feasible solution for the new model). Then, the new neighborhood is obtained increasing $t_{min}$, or setting it back to $1$ if $T$ has been reached (steps 16-19). 

\begin{algorithm}
	\caption{Matheuristic Improvement Procedure.}
	\label{Algo:MatheuristicAlgorithm}
	\begin{algorithmic}[1]
		\STATE INPUT: an initial solution $\bar{X}$, initial and final neighborhood size $T_{init}$ and $T_{final}$, model $M$ ($M_2^{BFTS}$ or $M_2^{PBFTS}$)
		\FOR {$T_{dim} = T_{init},\cdots, T_{final} $}
		\STATE Set $t_{min}=1, t_{impr} = T$
		\WHILE{$t_{min} \neq t_{impr}$}
		\STATE Consider model $M$ (without any additional constraint)
		\IF {$t_{min} + T_{dim} - 1 \leq T$}
		\STATE Add to $M$ constraints $x_{ti} = \bar{x}_{ti}$ for all $t$ such that $t < t_{min}$ or $ t > t_{min} + T_{dim} - 1$
		\ELSE
		\STATE Add to $M$ constraints $x_{ti} = \bar{x}_{ti}$ for all $t$ such that $ t_{min} + T_{dim} - T - 1 < t < t_{min}$
		\ENDIF
		\STATE $O_{last} = OF(\bar{X})$
		\STATE Optimally solve model $M$, retrieving the new solution $\bar{X}$
		\IF {$OF(\bar{X}) < O_{last}$}
		\STATE $t_{impr} = t_{min}$
		\ENDIF
		\STATE $t_{min} = t_{min} + 1$
		\IF {$t_{min} > T$}
		\STATE $t_{min} = 1$
		\ENDIF
		\ENDWHILE
		\ENDFOR
		\STATE OUTPUT: the improved solution $\bar{X}$
	\end{algorithmic}
\end{algorithm}

\section{Computational tests}
\label{SecTests}
We focused on grid graphs and generated instances according to the scheme proposed in \cite{MiHeMa}. We considered graphs with 25, 100, 225, 400, 900 and 1225 cells where each fuel age threshold $\sigma_i$ is randomly selected among 4, 8, or 12 years and each fuel age $a_i$ is randomly selected between 1 and 12 years. The time horizon $T$ is 10 years. Each cell is connected with three neighbouring cells by considering north-westerly prevailing wind direction, as indicated in \cite{MiHeMa}. Objective function weights and cells treatment costs are constant ($w_{tij} = c_{ti} = 1$ for all $t=1,\cdots,T$ and $i\in A$). Each budget value $b_t$ is equal to 5\% of the total treatment cost of all the cells, as considered in the most difficult instances in \cite{MiHeMa}.\\
A second instance type has been then generated, considering both weights and treatment costs as random integer numbers uniformly distributed in $[1,20]$.\\ 
For each of the two instance types and 6 landscape size we generated 10 instances, for a total of 120 different instances.
All tests have been run using MILP solver CPLEX 12.8 running on an Intel i5 CPU @ 3.0 GHz with 16 GB of RAM, within a time limit of 1800 seconds.\\

In the first computational tests, we tested on the non periodic version of the problem (BFTS) the model in \cite{MiHeMa} and our models $M_1^{BFTS}$ and $M_2^{BFTS}$. In Table \ref{tab:CPUTime1_1}, we report the performances of the models on the instances with constant costs and weights. For each landscape size, we report the average solution value (column ``Average Sol. Value"), the average CPU time in seconds (column ``Average Time"), and the number of instances, out of 10, solved to optimality with each model (column ``Opt"). In all tables, bold entries highlight the best average solution value for each landscape size. The results illustrate that the model proposed in \cite{MiHeMa} is not capable of solving to optimality any of the instances with 400 or more cells. Model $M_1^{BFTS}$ is able to solve to optimality all the instances up to 400 nodes, while model $M_2^{BFTS}$ exhibits much better performances and reaches all optimal solutions within the time limit. 
A possible explanation on the improved performances of the latter model is that in all instances the number of nonzero coefficients in model $M_2^{BFTS}$ is more than halved (on average) with respect to model $M_1^{BFTS}$.

\begin{table}[ht]
\scalebox{0.85}{
	\centering
	\scriptsize
	\begin{tabular}{|l|*{3}{c|}*{3}{c|c|}} 
		\hline
		&   \multicolumn{3}{|c|}{Model} 			& \multicolumn{3}{|c|}{Model}  & \multicolumn{3}{|c|}{Model} \\ 
		&   \multicolumn{3}{|c|}{in \cite{MiHeMa}} 	& \multicolumn{3}{|c|}{$M_1^{BFTS}$}  & \multicolumn{3}{|c|}{$M_2^{BFTS}$}\\ \hline
		
		\text{Landscape} & Average & Average & Opt & Average & Average & Opt & Average & Average & Opt \\															
		\text{Size}  &  Sol. Value & Time & & Sol. Value & Time & & Sol. Value & Time & \\	\hline
		
		25 (5 by 5)		&	\textbf{113.7}	&	1.2		&	10	&	\textbf{113.7}	&	0.0		&	10	&	\textbf{113.7}	&	0.0		&	10		\\ \hline
		100 (10 by 10)	&	\textbf{360.2}	&	40.9	&	10	&	\textbf{360.2}	&	3.1		&	10	&	\textbf{360.2}	&	1.3		&	10		\\ \hline
		225 (15 by 15)	&	925.7	&	770.5	&	7	&	\textbf{925.7}	&	68.7	&	10	&	\textbf{925.7}	&	7.3		&	10		\\ \hline
		400 (20 by 20)	&	1645.6	&	1800.0	&	0	&	\textbf{1640.5}	&	362.7	&	10	&	\textbf{1640.5}	&	38.2	&	10		\\ \hline
		900 (30 by 30)	&	4766.8	&	1800.0	&	0	&	3936.1	&	1800.0	&	0	&	\textbf{3809.9}	&	282.4	&	10		\\ \hline
		1225 (35 by 35)	&	6811.5	&	1800.0	&	0	&	6027.7	&	1800.0	&	0	&	\textbf{5257.4}	&	867.9	&	10		\\ \hline
		
	\end{tabular}
}
	\caption{BFTS instances with constant costs and weights.}
	\label{tab:CPUTime1_1}
\end{table}
The same tests have been repeated for the instances with variable costs and weights. The results are summarized in table \ref{tab:CPUTime1_2}. The model proposed in \cite{MiHeMa} (adapted, adding weights $w_{tij}$ in the objective function) is capable of solving to optimality all the instances only for the landscapes with $25$ and $100$ cells. The proposed models are able to solve all the instances up to 225 cells. Again, the best performing model is $M_2^{BFTS}$, able to solve all instances but 3 in the given time limit (one each for the 400, 900 and 1225 cells sets). Considering that 117 out of 120 instances are already optimally solved by the model, the matheuristic has not been tested on the non periodic case.\\

\begin{table}[ht]
\scalebox{0.85}{
	\centering
	\scriptsize
	\begin{tabular}{|l|*{3}{c|}*{3}{c|c|}} 
		\hline
		&   \multicolumn{3}{|c|}{Model} 			& \multicolumn{3}{|c|}{Model}  & \multicolumn{3}{|c|}{Model} \\ 
		&   \multicolumn{3}{|c|}{in \cite{MiHeMa}} 	& \multicolumn{3}{|c|}{$M_1^{BFTS}$}  & \multicolumn{3}{|c|}{$M_2^{BFTS}$}\\ \hline
		
		\text{Landscape} & Average & Average & Opt & Average & Average & Opt & Average & Average & Opt \\															
		\text{Size}  &  Sol. Value & Time & & Sol. Value & Time & & Sol. Value & Time & \\	\hline
		
25 (5 by 5)		&	\textbf{590.2}	&	0.0		&	10	&	\textbf{590.2}	&	0.0		&	10	&	\textbf{590.2}	&	0.0		&	10		\\ \hline
100 (10 by 10)	&	\textbf{1468.0}	&	270.2	&	10	&	\textbf{1468.0}	&	22.9	&	10	&	\textbf{1468.0}	&	27.7	&	10		\\ \hline
225 (15 by 15)	&	3323.3	&	1539.0	&	3	&	\textbf{3320.4}	&	189.5	&	10	&	\textbf{3320.4}	&	219.4	&	10		\\ \hline
400 (20 by 20)	&	6191.5	&	1800.0	&	0	&	\textbf{6175.1}	&	885.7	&	8	&	\textbf{6175.1}	&	695.0	&	9		\\ \hline
900 (30 by 30)	&	20430.9	&	1800.0	&	0	&	14756.9	&	1193.5	&	5	&	\textbf{14756.1}	&	1041.8	&	9		\\ \hline
1225 (35 by 35)	&	32824.4	&	1800.0	&	0	&	21027.3	&	1424.4	&	5	&	\textbf{21026.5}	&	1104.8	&	9		\\ \hline
	\end{tabular}
}
	\caption{BFTS instances with variable costs and weights.}
	\label{tab:CPUTime1_2}
\end{table}

We then considered the same 120 instances in the context of periodic planning and benchmarked the proposed matheuristic against model $M_2^{PBFTS}$. The CPU time limit is again 1800 seconds. 
The matheuristic have been tested with several configurations: in the initial solution we tested $K$ ranging from $1$ to $100$, while in the improvement procedure the interval $T_{init}-T_{final}$ was ranging, with different combinations, from $2$ to $6$. The results presented here refer to the settings $K=20$, $T_{init}=4$, $T_{end}=5$, which guarantee a good trade-off between running times and solution quality. \\

The results for the instances with constant costs and weights are reported in Table \ref{tab:CPUTime2_1}. In general, the presence of a periodic planning makes the instances harder to solve. In fact, the proposed model does not obtain all optimal solutions and is outperformed, in terms of solution quality,  by the matheuristic on large instances with 400, 900 and 1225 cells. \\

\begin{table}[ht]
	\centering
\begin{tabular}{|l|*{3}{c|}*{2}{c|c|}} 
  \hline
    &   \multicolumn{3}{|c|}{Model} & \multicolumn{2}{|c|}{Matheuristic}  \\ 
        &   \multicolumn{3}{|c|}{$M_2^{PBFTS}$} & \multicolumn{2}{|c|}{approach}  \\ \hline
								\text{Landscape} & Average & Average & Opt & Average & Average \\															
									\text{Size}  &  Sol. Value & Time & & Sol. Value & Time \\			\hline	
25 (5 by 5)		&	\textbf{166.1}	&	0.5		&10	& 166.5		&	0.1		\\ \hline
100 (10 by 10)	&	\textbf{519.5}	&	916.0	&7	& 520.2		&	6.4		\\ \hline
225 (15 by 15)	&	\textbf{1340.2}	&	1800.0	&0	& 1342.2	&	20.9		\\ \hline
400 (20 by 20)	&	2347.9	&	1800.0	&0	& \textbf{2347.4}	&	64.9		\\ \hline
900 (30 by 30)	&	7151.8	&	1800.0	&0	& \textbf{5426.6}	&	360.6		\\ \hline
1225 (35 by 35)	&	9035.3	&	1800.0	&0	& \textbf{7483.9}	&	775.3		\\ \hline

\end{tabular}
		\caption{PBFTS instances with constant costs and weights.}
		\label{tab:CPUTime2_1}
\end{table}

Table \ref{tab:CPUTime2_2} presents the results for the periodic planning on the instances with variable costs and weights. Again, the model is not able to optimally solve the larger instances within the time limit, and the matheuristic approach outperforms it in terms of solution quality for the larger landscapes (900 and 1225 cells).

\begin{table}[ht]
	\centering
	\begin{tabular}{|l|*{3}{c|}*{2}{c|c|}} 
		\hline
		&   \multicolumn{3}{|c|}{Model} & \multicolumn{2}{|c|}{Matheuristic}  \\ 
		&   \multicolumn{3}{|c|}{$M_2^{PBFTS}$} & \multicolumn{2}{|c|}{approach}  \\ \hline
		\text{Landscape} & Average & Average & Opt & Average & Average \\															
		\text{Size}  &  Sol. Value & Time & & Sol. Value & Time \\			\hline	
25 (5 by 5)		&	\textbf{849.4}		&	0.1		&10	& 850.1		&	0.2		\\ \hline
100 (10 by 10)	&	\textbf{2399.1}		&	1052.5	&5	& 2409.9	&	9.7		\\ \hline
225 (15 by 15)	&	\textbf{5754.9}		&	1647.1	&1	& 5769.1	&	38		\\ \hline
400 (20 by 20)	&	\textbf{10822.8}		&	1800.0	&0	& 10853.4	&	105.1		\\ \hline
900 (30 by 30)	&	25347.8		&	1800.0	&0	& \textbf{25077.5}	&	490.9		\\ \hline
1225 (35 by 35)	&	45109.6		&	1800.0	&0	& \textbf{35241.5}	&	1043.8		\\ \hline

	\end{tabular}
	\caption{PBFTS instances with variable costs and weights.}
	\label{tab:CPUTime2_2}
\end{table}

\section{Conclusions and future directions}
\label{Sec:concl}
In this work, we proposed improved mathematical models and a matheuristic approach for the  Budget constrained Fuel Treatment Scheduling Problem, a well-known optimization problem in fuel management. We evaluated the effectiveness of the developed models and algorithms on a large set of instances from the literature and for a problem variant with periodic planning. 

In future research, it would be worthy to consider further generalizations of the problem, in order to identify better fits with real-life scenarios. In this respect,  extensions of the ILP models proposed in \cite{MiHeMa} were considered in \cite{leReHe}, \cite{MiHeMa}, \cite{RaOzHeRe} and \cite{16RaOzHeRe}. 
Interestingly, our approaches can be easily adapted to most of those extensions as mentioned in the Introduction. 
In addition, requirements to preserve the fauna habitat could be taken into account in deriving fuel management strategies. For instance, an area can be treated in a given period only if the present fauna could move to a sufficient number of adjacent areas with a suitable habitat. 
Finally, here we sticked to the use of landscapes divided into a grid of square cells  as assumed in the work of \cite{MiHeMa} but our approach is extendable to any landscape that can fit into a network representation.

\section*{Acknowledgement}
This work was funded by the GEO-SAFE project and the EU Horizon2020 RISE programme,  grant agreement No 691161.


\begin{thebibliography}{1} 
\bibitem{Ball11} Ball, M. O.: Heuristics based on mathematical programming. Surveys in Operation Research and Management Science, 16, 21--38 (2011).
\bibitem{BDCG15}
Billaut, J.-C., Della Croce, F., Grosso, A.:
A single machine scheduling problem with two-dimensional vector packing constraints. European Journal of Operational Research, 243 (1), 75--81 (2015).  
\bibitem{BMMRG19}
Bhuiyan, T.H., Moseley, M.C., Medal, H.R., Rashidi, E., Grala, R.K.:
A stochastic programming model with endogenous uncertainty for incentivizing fuel reduction treatment under uncertain landowner behavior European Journal of Operational Research, 277 (2), 699--718 (2019).
\bibitem{BOAL09} Boer, M. M., Sadler, R. J., Wittkuhn, R. S., McCaw, L., Grierson, P. F.: Long-term impacts of prescribed burning on regional extent and incidence of wildfires-evidence from 50 years of active fire management in sw australian forests. Forest Ecology and Management, 259(1), 132--142 (2009).
\bibitem{DCGS13} Della Croce, F., Grosso, A., Salassa, F.: Matheuristics: embedding MILP solvers into heuristic algorithms for combinatorial optimization problems. In P. Siarry (Ed.), Heuristics: theory and applications, Nova Science Publishers, 31--52 (2013).
\bibitem{DCGS14} Della Croce, F., Grosso, A., Salassa, F.:
A matheuristic approach for the two-machine total completion time flow shop problem. Annals of Operations Research, 213 (1), 67--78 (2014). 
\bibitem{DCS14} Della Croce, F., Salassa, F.:
A variable neighborhood search based matheuristic for nurse rostering. Annals of Operations Research, 218 (1), 185--199 (2014). 
\bibitem{DCGS19}
Della Croce, F., Grosso, A., Salassa, F.:
Minimizing total completion time in the two-machine no-idle no-wait flow shop problem. Journal of Heuristics, forthcoming, doi:10.1007/s10732-019-09430-z.
\bibitem{DeTa15} Demange, M., Tanasescu, C.: A Graph Approach for Fuel Treatment Scheduling. Working paper, RMIT University (2015).
\bibitem{CG1} Desaulniers, G., Desrosiers, J., Solomon, M.M.: Column generation. Springer US (2005).
\bibitem{DNV18}
Doi, T.,  Nishi, T., Voss, S.:
Two-level decomposition-based matheuristic for airline crew rostering problems with fair working time.
European Journal of Operational Research, 267 (2), 428--438 (2018).
\bibitem{FPR17}
Fanjul-Peyro, L., Perea, F., Ruiz, R.:
Models and matheuristics for the unrelated parallel machine scheduling problem with additional resources. European Journal of Operational Research, 260(2), 482--493 (2017).
\bibitem{FEBO03} Fernandes, P. M., Botelho, H. S.: A review of prescribed burning effectiveness in fire hazard reduction. International Journal of Wildland Fire, 12(2), 117--128 (2003).
\bibitem{GJ82} Garey, M.R., Johnson, D.S.: {\em 
Computers and Intactability: A Guide to the Theory of NP-Completeness}. Freeman and CO., New York (1982).
\bibitem{GMP17} Gillen, C.P., Matsypura, D., Prokopyev, O.A.: 
Operations Research Techniques in Wildfire Fuel Management. Chapter in: Springer Optimization and Its Applications book series (SOIA). 130,
119--135, (2017).
\bibitem{leReHe} Leon, J., Reijnders, V.M.J.J. , Hearne, J.W., Ozlen,  M., Reinke, K.J.: A Landscape-Scale Optimisation Model to Break the Hazardous Fuel Continuum While Maintaining Habitat Quality. Environmental Modelling and Assessment, 24, 369--379 (2019).
\bibitem{MLGPP19}
Macrina, G., Laporte, G., Guerriero, F., Di Puglia Pugliese, L.:
An energy-efficient green-vehicle routing problem with mixed vehicle fleet, partial battery recharging and time windows. European Journal of Operational Research, 276(3), 971--982 (2019).
\bibitem{MABRT17}
Martinez-Sykora, A., Alvarez-Valdes, R., Bennell, J.A., Ruiz, R., Tamarit, J.M.: Matheuristics for the irregular bin packing problem with free rotations. European Journal of Operational Research, 258(2), 440--455 (2017).
\bibitem{MPZ18} Matsypura, D., Prokopyev, O.A., Zahar, A.: Wildfire fuel management: network-based models and optimization of prescribed burning. European Journal of Operational Research, 264, 77--796 (2018).
\bibitem{MHH12} Minas, J.P., Hearne, J.W., Handmer, J.W.: A review of operations research methods applicable to wildfire management. International Journal of Wildland Fire, 21 (3), 189--196 (2012).
\bibitem{MiHeMa} Minas, J.P., Hearne, J.W., Martell, D.L.: A spatial optimisation model for multi-period landscape level fuel management to mitigate wild reimpacts. European Journal of Operational Research, 232, 412--422 (2014).
\bibitem{RaOzHeRe} Rachmawati, R., Ozlen, M., Hearne, J.W., Reinke, K.J.: Fuel treatment planning: Fragmenting high fuel load areas while maintaining availability and connectivity of faunal habitat. Applied Mathematical Modelling, 54, 298--310 (2018). 
\bibitem{16RaOzHeRe} Rachmawati, R., Ozlen, M., Reinke, K. J., Hearne, J. W.: An optimisation approach for fuel treatment planning to break the connectivity of high-risk regions. Forest Ecology and Management, 368, 94--104 (2016).
\bibitem{SAS20}
Shahmanzari, M., Aksen, D., Salhi, S.: 
Formulation and a two-phase matheuristic for the roaming salesman problem: Application to election logistics. European Journal of Operational Research, 280(2), 656--670 (2020).
\bibitem{WRK08} Wei, Y., Rideout, D., Kirsch, A.: An optimization model for locating fuel treatments across a landscape to reduce expected fire losses. Canadian Journal of Forest Research, 38 (4), 868--877 (2008).
\end{thebibliography}
\end{document}